\documentclass[letterpaper, 10pt, conference]{ieeeconf}
\usepackage{graphicx} 

\IEEEoverridecommandlockouts

\def\usingarxiv{false}

\usepackage{multibib}
\newcites{app}{Appendix References}

\usepackage{amsmath, cases}
\usepackage{subfiles,balance}
\usepackage{amssymb}
\usepackage{amsthm}
\usepackage{algorithm}
\usepackage{algpseudocode}
\usepackage{etoolbox} 
\usepackage[hidelinks]{hyperref} 

\newcommand{\arxiv}[2]{\ifbool{\usingarxiv}{#1}{#2}} 

{
    \theoremstyle{plain}
    \newtheorem{assumption}{Assumption}
    
    \newtheorem{theorem}{Theorem}
    \newtheorem{lemma}{Lemma}
    
    \newtheorem{corollary}{Corollary}
    \newtheorem{definition}{Definition}
}

\makeatletter
\newcommand\fs@betterruled{%
  \def\@fs@cfont{\bfseries}\let\@fs@capt\floatc@ruled
  \def\@fs@pre{\vspace*{7pt}\hrule height.8pt depth0pt \kern2pt}%
  \def\@fs@post{\kern2pt\hrule\relax}%
  \def\@fs@mid{\kern2pt\hrule\kern2pt}%
  \let\@fs@iftopcapt\iftrue}
\floatstyle{betterruled}
\restylefloat{algorithm}
\makeatother

\title{\LARGE \bf Stability of Control Lyapunov Function Guided Reinforcement Learning}
\author{
Zachary Olkin*, William D. Compton*, Aaron D. Ames
\thanks{* denotes equal contribution. The authors are with the Department of Control and Dynamical Systems at the California Institute of Technology. This research is supported by Technology Innovation Institute (TII).}
}

\begin{document}
 
\maketitle

\thispagestyle{empty}
\pagestyle{empty}

\renewcommand\qedsymbol{$\blacksquare$}

\begin{abstract}
    Reinforcement learning (RL) has become the de facto method for achieving locomotion on humanoid robots in practice, yet stability analysis of the corresponding control policies is lacking. Recent work has attempted to merge control theoretic ideas with reinforcement learning through control guided learning. A notable example of this is the use of a control Lyapunov function (CLF) to synthesize the reinforcement learning rewards, a technique known as CLF-RL, which has shown practical success. This paper investigates the stability properties of optimal controllers using CLF-RL with the goal of bridging experimentally observed stability with theoretical guarantees. The RL problem is viewed as an optimal control problem and exponential stability is proven in both continuous and discrete time using both core CLF reward terms and the additional terms used in practice. The theoretical bounds are numerically verified on systems such as the double integrator and cart-pole. Finally, the CLF guided rewards are implemented for a walking humanoid robot to generate stable periodic orbits.
\end{abstract}

\section{Introduction}
Control of humanoid robots has progressed exponentially in the last few years, mostly due to advances in applied Reinforcement Learning (RL). Although bipedal and humanoid robots have been studied through the lens of control theory for decades including works on reduced order models \cite{kajita_3d_2001, pratt_capture_2006}, optimization based control \cite{wensing_optimization_2024,ames_rapidly_2014}, and the hybrid zero dynamics (HZD) methodology \cite{westervelt_hybrid_2003, westervelt_feedback_2018}, RL is emerging as the de-facto control methodology for these systems \cite{gu_evolution_2025}. The robustness and ease of deployment that accompanies RL has shown to be unmatched in practical terms.

Yet the practical success of RL is only part of the story. In general, these methods are applied to robot systems without strong theoretical groundings. In this paper, we start to bridge the gap between theory and practice for RL controllers that have found success when applied experimentally to humanoid robots. Specifically, we examine the theoretical properties of CLF-RL, a control-theoretic reward shaping method leveraging control Lyapunov functions (CLFs). This method has already successfully achieved locomotion on humanoids operating in real, outdoor environments \cite{li_clf-rl_2026, olkin_chasing_2025}.

CLFs are a classic tool in nonlinear control theory for synthesizing certifiably stable controllers for nonlinear systems \cite{sontag1989universal}. These controllers have been used in the context of bipedal locomotion successfully in the past \cite{ames_rapidly_2014, galloway_torque_2015}. CLFs provide a set of control actions, rather than a single fixed action, to be taken at any moment in time, and therefore an optimal action can be chosen out of this set. Given this, CLF-RL embeds both the Lyapunov function and the associated stability (decreasing) condition into the reward, thus reducing the number of rewards and incentivizing the RL to take certifiably stable actions. It should be noted that the RL is not restricted to taking the action of the CLF, and therefore can take other actions if they prove to be more optimal or stable in the long run. We find that in practice, this allows the closed loop controller to gain additional stability properties beyond what a CLF-QP controller \cite{ames_rapidly_2014} might yield.

\begin{figure}
    \centering
    \includegraphics[width=1.0\linewidth]{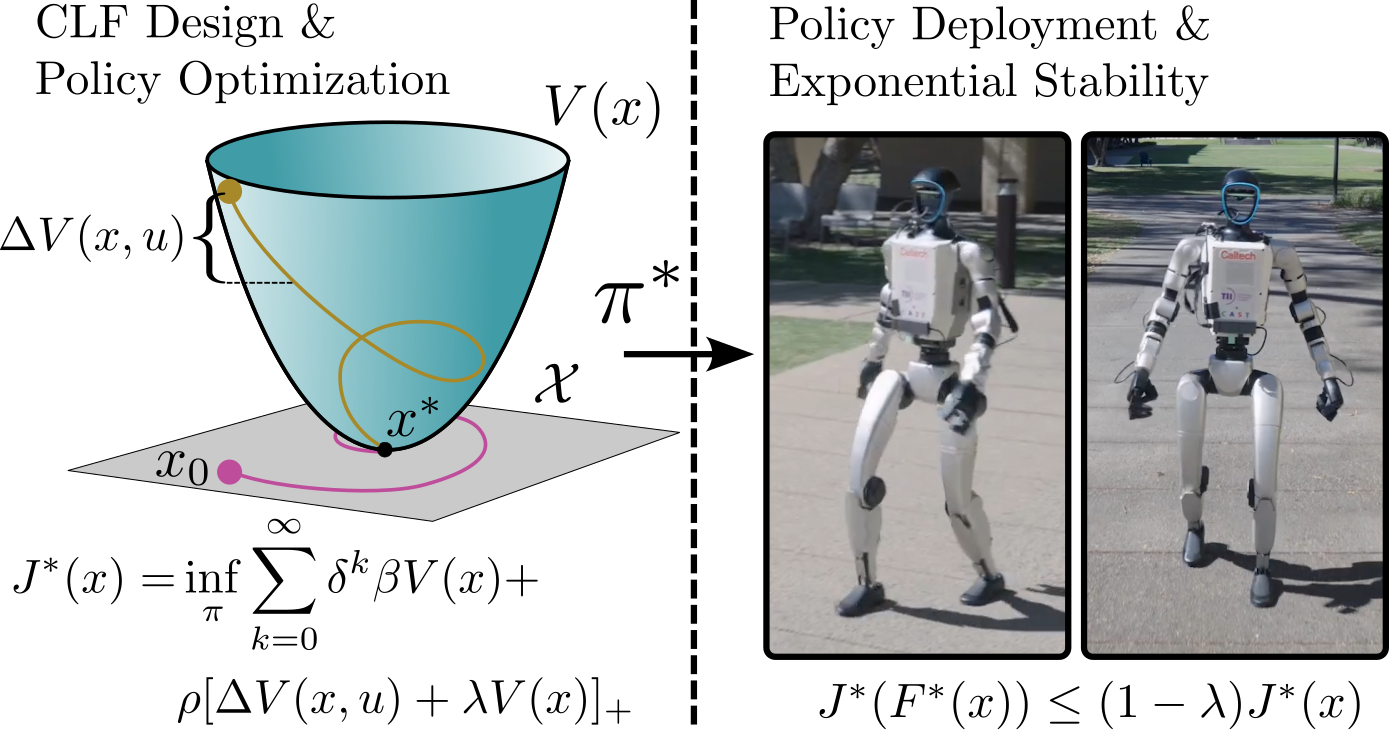}
    \caption{The main ideas behind CLF guided RL. A CLF is designed offline and used in an RL/optimal control problem as the reward/cost. Then the optimal policy, $\pi^*$ is applied to the system resulting in provable exponential stability and in stable humanoid locomotion.}
    \label{fig:clf_rl_theory_hero_fig}
    \vspace{-5mm}
\end{figure}

In applied RL for dynamic humanoid motions, the standard has become tracking reference trajectories, which commonly comes from retargeted human data \cite{sleiman_zest_2026, liao_beyondmimic_2025}. The classic tracking reward penalizes the normed distance to the trajectory, generally with a different reward for each type of output, such as joint positions, or frame velocities. On the other hand, in CLF-RL a dynamically feasible trajectory is used and the reward is shaped with a CLF. These CLF rewards have shown improved tracking capability over the standard tracking rewards \cite{olkin2026chasingautonomydynamicretargeting}.  While other works on RL have used inspiration from Lyapunov stability theory \cite{dong_principled_2020}, they do not actually use a true CLF for the system, instead applying the stability ideas from a CLF to arbitrary rewards. On the other hand \cite{westenbroek_lyapunov_2022} uses a CLF in the cost, but only uses this to fine-tune and augment other controllers and only proves asymptotic stability, not exponential. 

The goal of this paper is to provide theoretic grounding for CLF-RL.  
To this end, we view the RL problem as an optimal-control problem, and prove that CLF cost shaping leads to certifiably exponentially stable controllers. We start by considering the continuous time optimal control problem with CLF-related rewards, then progress to the discrete time problem, and lastly progress to include additional practical rewards. In each case we demonstrate that exponential stability to the origin is achieved when an optimal policy is found, thus taking the first steps to ground CLF-RL in principled theory. Further, we prove that the resulting controller is robust in the discrete time case. We demonstrate the stability and theoretical bounds through numerical examples on the double integrator and cart pole system. Lastly, we apply the method for tracking a periodic orbit on a humanoid robot (i.e. a walking gait) through RL and demonstrate that we can achieve accurate reference tracking and stable walking.

\section{Preliminaries}
We consider problem formulations and CLFs using both continuous and discrete time nonlinear dynamics.
\subsection{Continuous Time}
In continuous time, we consider dynamics of the form 
\begin{equation}
\label{eq:cont_dynamics}
    \dot{x} = f(x) + g(x)u
\end{equation}
where $x \in \mathbb{R}^n$, and $u \in \mathcal{U} \subset \mathbb{R}^m$ with $f(x^*) = 0$ and $0 \in \text{int}\:\mathcal{U}$. Here $x^*$ denotes the target equilibrium which we take as $x^* = 0$ without loss of generality.

Control Lyapunov functions (CLFs) will be used throughout. Let $V : \mathcal{X} \rightarrow \mathbb{R}_{\geq 0}$ denote a given CLF on a region $\mathcal{X} \subset \mathbb{R}^n$. To be a CLF, $V$ must be positive definite, quadratically bounded on $\mathcal{X}$, and also satisfy, for $\alpha > 0$:
\begin{equation*}
    \inf_{u \in \mathcal{U}} \nabla V(x)^T(f(x) + g(x)u) \leq -\alpha V(x).
\end{equation*}

\subsection{Discrete Time}
Here we consider discrete time dynamics of the form
\begin{equation*}
    x_{k+1} = F(x_k, u_k)
\end{equation*}
with $x_k \in \mathbb{R}^n$ and $u_k \in \mathcal{U} \subset \mathbb{R}^m$. $F: \mathcal{X} \times \mathcal{U} \rightarrow \mathcal{X}$ is continuous, $F(x^*,0) = 0$, and $\mathcal{U}$ is compact.

To be a discrete time CLF, $V(x)$ must similarly be positive definite, quadratically bounded, and satisfy
\begin{equation*}
   \inf_{u \in \mathcal{U}} V(F(x, u)) - V(x) \leq -\alpha V(x)
\end{equation*}
for $\alpha \in (0, 1]$.
\section{Continuous Time Stability}

\begin{assumption}
\label{assump:clf_existance}
    There exist a CLF $V(x)$ with constants $c_1$, $c_2$, $\alpha > 0$ and a corresponding admissible feedback controller $\mu : \mathcal{X} \rightarrow \mathcal{U}$ such that for all $x \in \mathcal{X}$,
    \begin{align}
    \label{eq:clf_cond}
        c_1 \| x \|^2 &\leq V(x) \leq c_2 \|x\|^2, \\
        \nabla V(x)^T(f(x) &+ g(x)\mu(x)) \leq -\alpha V(x). \nonumber
    \end{align}
\end{assumption}

This assumes the existence of a feedback controller stabilizing the system at the given convergence rate. Now we wish to optimize a separate controller (ultimately in order to use the powerful techniques of RL) where the cost encourages behavior that stabilizes the system at least as fast as the CLF. Since there is an entire set of feasible inputs, there are in general many choices of control action that satisfy this condition. Looking towards RL, we put this in the cost, rather than a constraint, and thus we do not enforce Lyapunov style convergence explicitly. As we will show, the resulting optimal controller is still exponentially stabilizing.

To design the new optimal controller fix design parameters $0 < \lambda \leq \alpha$, $\beta > 0$, $\rho > 0$, and $\gamma > 0$ where $\lambda$ is the desired decay rate in the CLF cost, $\beta$ is the weight on the Lyapunov cost, $\rho$ the weight on the CLF-violation penalty, and $\gamma$ the discount factor. We consider the following stage cost
\begin{equation}
\label{eq:stage_cost}
    \ell(x,u) := \beta V(x) + \rho [\dot{V}(x,u) + \lambda V(x)]_+,
\end{equation}
with the corresponding discounted infinite horizon optimal-control problem:
\begin{equation*}
    J_\pi(x_0) = \int_0^\infty e^{-\gamma t} \ell(x^\pi(t), \pi(x^\pi(t))) dt
\end{equation*}
where $x^\pi$ represents the flow of \eqref{eq:cont_dynamics} under the action of $u = \pi(x)$. The corresponding value function is then
\begin{equation*}
    J^*(x) = \inf_\pi J_\pi(x).
\end{equation*}

\begin{figure*}
    \centering
    \includegraphics[width=1.0\linewidth]{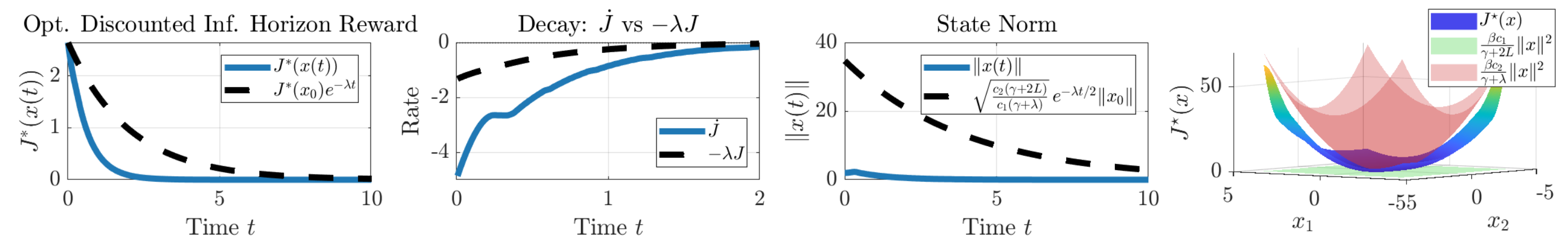}
    \caption{Theoretical bounds plotted against the numerical solution to the optimal control problem for the continuous time double integrator with the CLF-RL costs. We can see that theoretical bounds are satisfied and the state exponentially converges to the origin.}
    \label{fig:ct_double_integrator_bounds}
    \vspace{-3mm}
\end{figure*}

\begin{assumption}
\label{assump:optimal_regularity}
    The optimal value function $J^*$ is finite on $\mathcal{X}$ and the optimal control problem admits a minimizing control policy $\pi^*: \mathcal{X} \rightarrow \mathcal{U}$. The resulting closed-loop vector field is
    \begin{equation}
    \label{eq:optimal_cl_sys}
        f^*(x) := f(x) + g(x)\pi^*(x)
    \end{equation}
\end{assumption}

We start by using the principle of optimality as to break the value function down into two parts: the cost over a small time interval $h$, and the optimal cost:
\begin{equation*}
    J^*(x) = \inf_{u \in \mathcal{U}} \left\{ \int_0^h e^{-\gamma t} \ell(x(t), u(t)) dt + e^{-\gamma h}J^*(x(h))\right\}.
\end{equation*}
Linearization and simplification yields the discounted HJB equation (see \cite{primbs_nonlinear_1999} for details):
\begin{equation}
\label{eq:hjb_inf}
    \gamma J^*(x) = \inf_{u \in \mathcal{U}} \left\{ \ell(x,u) + \nabla J^*(x)^T(f(x) + g(x)u) \right\}
\end{equation}

Consider the following sets defined for convenience. Let $S_d$ denote the $d$ sublevel set of the value function:
\begin{equation*}
    S_d = \{x \in \mathcal{X} : J^*(x) \leq d\}.
\end{equation*}
Let $\Omega_c$ denote the $c$ sublevel set of the Lypaunov function:
\begin{equation*}
    \Omega_c = \{x \in \mathcal{X} : V(x) \leq c\}.
\end{equation*}
\begin{theorem}
\label{thm:cont_time_exp_stab}
    Suppose assumptions \ref{assump:clf_existance} and \ref{assump:optimal_regularity} hold.
    Additionally, assume that $f^*$ is locally lipschitz continuous.
    Then there exists constants $d > 0$, $c > 0$, and $L > 0$ such that $S_d$ is compact and forward invariant under the action of $\pi^*$, and further, for all $x \in \Omega_c$:
    \begin{align}
    \label{eq:J_clf_bounds}
        \frac{\beta c_1}{\gamma + 2L} \|x\|^2 \leq J^*(x) \leq \frac{\beta c_2}{\gamma + \lambda} \|x\|^2 \\
        \dot{J}^*(x) \leq -\lambda J^*(x)
    \end{align}
    and therefore $x^* = 0$ is locally exponentially stable under the optimal CLF-RL policy $\pi^*$. In particular, the solution $x(t)$ to $\dot{x}(t) = f^*(x)$ satisfies:
    \begin{equation}
    \label{eq:ct_clf_convergence_state}
        \|x(t)\| \leq \sqrt{\frac{c_2(\gamma + 2L)}{c_1(\gamma + \lambda)}}e^{-\lambda t/2} \|x_0\| \quad \forall t \geq 0.
    \end{equation}
\end{theorem}
The given optimal control problem, inspired by CLF-based rewards in the context of RL, admits an optimal controller which is exponentially stabilizing. To prove this, we make use of the following four lemmas.

\begin{lemma}
\label{lem:J_positive_definite}
    The value function $J^*$ is positive definite on $\mathcal{X}$: $J^*(0) = 0$, $J^*(x) > 0 \; \forall x \neq 0$.
\end{lemma}
\begin{proof}
    Since $V(x) \geq 0$ and $[\cdot]_+ \geq 0$, we have that $l(x, u) \geq 0$ and therefore $J^*(x) \geq 0$.
    The fact that $J^*(0) = 0$ follows immediately from $V(0) = 0$ and $J^* \geq 0$.
    It remains to show that $J^*(x) = 0$ implies $x = 0$. Let $x \in \mathcal{X}$ and suppose $J^*(x) = 0$. The value is obtained at the optimal feedback $\pi^*$, so along the corresponding optimal trajectory $x_{\pi^*}(t)$:
    \begin{equation*}
        0 = \int_0^\infty e^{-\gamma t} l\big(x_{\pi^*}(t), \pi^*(x_{\pi^*}(t))\big) dt
    \end{equation*}
    Since the integrand is nonnegative it must vanish almost everywhere, so $V(x_{\pi^*}(t)) = 0$ for almost every $t \geq 0$. Since $x_{\pi^*}(t)$ is continuous and $V(x)$ is continuous, the function $V(x_{\pi^*}(t))$ is continuous. A continuous nonegative function that is zero almost everywhere must be identically zero. Hence $V(x_{\pi^*}(t)) \equiv 0 \; \forall t \geq 0$. Since $V(x)$ is positive definite then $x_{\pi^*} \equiv 0 \; \forall t$. Therefore $J^*(x) > 0$ for every $x \neq 0$.
\end{proof}

\begin{lemma}
\label{lem:upper_bound_J}
    For all $x \in \mathcal{X}$,
    \begin{equation}
    \label{eq:J_start_lemma_bounds}
        J^*(x) \leq \frac{\beta}{\gamma + \lambda}V(x),
    \end{equation}
\end{lemma}
\begin{proof}
    To obtain the upper bound we consider the CLF controller, $\mu(x)$ from Assumption \ref{assump:clf_existance},
    \begin{equation*}
    \nabla V(x) (f(x) + g(x) \mu(x)) \leq -\alpha V(x) \leq -\lambda V(x)
    \end{equation*}
    where $\lambda \leq \alpha$ by construction. Therefore under the action of $\mu(x)$ the term $[\dot{V}(x, \mu(x)) + \lambda V(x)]_+ \equiv 0$. Then from the comparison lemma,
    \begin{equation*}
        V(x_\mu(t)) \leq e^{-\alpha t}V(x) \leq e^{-\lambda t}V(x).
    \end{equation*}
    From this, we upper bound the value function
    \begin{align*}
        J^*(x) &\leq J_\mu(x) = \int_0^\infty e^{-\gamma t}\beta V(x_\mu(t)) dt \\
        &\leq \beta \int_0^\infty e^{-(\gamma + \lambda)t}dtV(x) = \frac{\beta}{\gamma + \lambda}V(x)
    \end{align*}
    and \eqref{eq:J_start_lemma_bounds} is shown.
\end{proof}

\begin{figure*}
\vspace{2mm}
    \centering
    \includegraphics[width=1.0\linewidth]{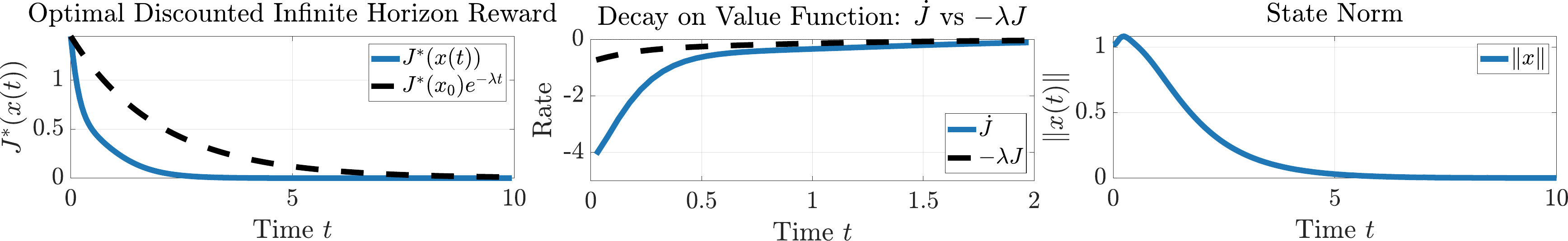}
    \caption{The optimal control of the continuous time cart-pole system using the CLF-RL costs plotted with the theoretical bounds. The cart-pole system exponentially stabilizes to the origin under the optimal control. No state norm bound is given as the Lipschitz constant of the optimal closed loop system is not computed.}
    \label{fig:ct_cart_pole_bounds}
    \vspace{-5mm}
\end{figure*}

\begin{lemma}
\label{lem:J_decreasing}
    Let $x(t)$ solve the optimal closed loop system \eqref{eq:optimal_cl_sys}. Then
    \begin{equation}
    \label{eq:value_fcn_dec}
        \dot{J}^*(x) \leq -\lambda J^*(x)
    \end{equation}
    and
    \begin{equation}
    \label{eq:value_fcn_traj_bound}
        J^*(x(t)) \leq e^{-\gamma t}J^*(x_0) \quad \forall t \geq 0.
    \end{equation}
\end{lemma}
\begin{proof}
    By \eqref{eq:hjb_inf} the minimizing control, $\pi^*$, satisfies
    \begin{equation*}
        \gamma J^*(x(t)) = \ell(x, \pi^*(x)) + \nabla J^*(x)^Tf^*(x).
    \end{equation*}
    Then along the closed-loop trajectories, the chain rule yields
    \begin{equation}
    \label{eq:cl_value_fcn_derivative}
        \frac{d}{dt}J^*(x(t)) = \gamma J^*(x(t)) - \ell(x(t), \pi^*(x(t))).
    \end{equation}
    Since $\ell(x, \pi^*(x)) \geq \beta V(x)$ by \eqref{eq:stage_cost}, Lemma \ref{lem:upper_bound_J} implies
    \begin{equation}
    \label{eq:stage_lower_bound}
        \ell(x, \pi^*(x)) \geq \beta V(x) \geq (\gamma + \lambda)J^*(x).
    \end{equation}
    Then we can substitute \eqref{eq:stage_lower_bound} into \eqref{eq:cl_value_fcn_derivative} to get \eqref{eq:value_fcn_dec}. Applying comparison lemma on $\eqref{eq:value_fcn_dec}$ yields \eqref{eq:value_fcn_traj_bound}.
\end{proof}

\begin{lemma}
\label{lem:cont_time_asymp_stab}
    The optimal policy, $\pi^*$ renders the origin locally asymptotically stable. Moreover, $J^*$ is a Lyapunov function for the optimal closed loop system and decays exponentially according to \eqref{eq:value_fcn_traj_bound}.
\end{lemma}
\begin{proof}
    From Lemma \ref{lem:J_positive_definite}, $J^*$ is positive definite. By Lemma \ref{lem:J_decreasing}, $\dot{J}^*(x) \leq -\lambda J^*(x)$. Therefore $J^*$ is a Lyapunov function.

    Choose $r > 0$ so that the closed ball $\bar{B}_r \subset \mathcal{X}$, and define
    \begin{equation}
    \label{eq:min_on_sphere}
        m_r := \min_{\|x\|=r}J^*(x) > 0.
    \end{equation}
    Then pick $d \in (0, m_r)$. The sublevel set
    \begin{equation*}
        S_d := \{x \in \mathcal{X} : J^*(x) \leq d\}
    \end{equation*}
    is contained in $B_r$ and is forward invariant by Lemma $\ref{lem:J_decreasing}$. Since $J^*$ is continuous and positive definite on $\bar{B}_r$, there exist class-$\mathcal{K}$ functions $\underline{\alpha}_r$ and $\bar{\alpha}_r$ such that
    \begin{equation*}
        \underline{\alpha}_r(\|x\|) \leq J^*(x) \leq \bar{\alpha}_r(\|x\|) \quad \forall x \in \bar{B}_r.
    \end{equation*}
    Thus for every $x_0 \in S_d$ the flow of the system under the action of the optimal controller satisfies
    \begin{equation*}
        \|x(t)\| \leq \underline{\alpha}_r^{-1}(e^{-\lambda t}\bar{\alpha}_r(\|x_0\|)), \quad \forall t \geq 0.
    \end{equation*}
    demonstrating asymptotic stability.
\end{proof}

Lemma \ref{lem:cont_time_asymp_stab} proves that the resulting controller is asymptotically stable without requiring Lipschitz continuity of the closed loop dynamics. Under this additional assumption, Theorem \ref{thm:cont_time_exp_stab} asserts exponential stability.  

\begin{proof}[Proof of Theorem \ref{thm:cont_time_exp_stab}]
    Let $r > 0$ such that $\bar{B}_r \subset \mathcal{X}$. From Lemma \ref{lem:cont_time_asymp_stab}, $S_d \subset B_r$ and is compact. It is also forward invariant as $J^*(x(t))$ is non-increasing along the optimal closed loop dynamics. As $S_d$ is a neighborhood of the origin and $V$ is continuous with $V(0) = 0$, there exists a $c > 0$ such that the sublevel set $\Omega_c = \{x \in \mathcal{X} : V(x) < c\}$ satisfies
    \begin{equation}
    \label{eq:clf_sublevel_in_sd}
        \Omega_c \subset S_d.        
    \end{equation}

    Since $f^*$ is locally Lipschitz by assumption and $S_d$ is compact, there exists $L > 0$ such that $\|f^*(x)\| \leq L\|x\| \; \forall x \in S_d$ where $L$ is the Lipschitz constant of $f^*$.

    Now fix $x_0 \in \Omega_c$ and let $x(t)$ denote the optimal closed loop trajectory. Then from \eqref{eq:clf_sublevel_in_sd} and the forward invariance of $S_d$, $x(t) \in S_d \; \forall t \geq 0$. Denote $r(t) := \|x(t)\|$ then the upper right Dini derivative satisfies:
    \begin{equation*}
        D^+r(t) \geq - \|\dot{x}(t)\| = -\|f^*(x(t))\| \geq -Lr(t),
    \end{equation*}
    where the Lipschitz bound was used for the last step. Then by the comparison lemma:
    \begin{equation*}
        \|x(t)\| \geq e^{-Lt}\|x_0\| \quad \forall t \geq 0
    \end{equation*}
    Now we can get the quadratic lower bound on $J^*$:
    \begin{align*}
        J^*(x_0) &= J_{\pi^*}(x_0)\geq \int_0^\infty e^{-\gamma t}\beta V(x(t))dt \\
        &\geq \beta c_1 \int_0^\infty e^{-\gamma t}\|x(t)\|^2 dt \\
        &\geq \beta c_1 \int_0^\infty e^{-(\gamma + 2L)t}dt \|x_0\|^2 \\
        &= \frac{\beta c_1}{\gamma + 2L}\|x_0\|^2.
    \end{align*}
    The upper quadratic bound follows from Lemma \ref{lem:upper_bound_J} and \eqref{eq:clf_cond}:
    \begin{equation*}
        J^*(x_0) \leq \frac{\beta}{\gamma + \lambda}V(x_0) \leq \frac{\beta c_2}{\gamma + \lambda}\|x_0\|^2
    \end{equation*}
    Finally, Lemma \ref{lem:J_decreasing} gives \eqref{eq:value_fcn_traj_bound} which can be combined with the quadratic bounds in \eqref{eq:J_clf_bounds} to yield the following:
    \begin{align*}
        \|x(t)\|^2 &\leq \frac{\gamma + 2L}{\beta c_1} J^*(x(t)) \leq \frac{\gamma + 2L}{\beta c_1}e^{-\lambda t}J^*(x_0) \\
        &\leq \frac{\gamma + 2L}{\beta c_1}e^{-\lambda t}\frac{\beta c_2}{\gamma + \lambda}\|x_0\|^2.
    \end{align*}
    Taking a square root results in \eqref{eq:ct_clf_convergence_state}. Therefore the origin is locally exponentially stable under the optimal policy $\pi^*$.
\end{proof}

We can apply these concepts on simple examples to verify the claims. Specifically, we use the double integrator and the cart pole systems. Fig. \ref{fig:ct_double_integrator_bounds} shows the results of numerical methods applied to the optimal control problem for the double integrator. We plot the theoretical bounds and show that closed loop trajectories satisfy these bounds. To get the numerical solutions, the HJB equation is solved iteratively over a grid in state space using the semi-Lagrangian method \cite{falcone_semi-lagrangianapproximation_2014}. This approach approximates the value function numerically, allowing us to verify quadratic bounds on the value function. The Lipschitz constant of the closed loop system used in the state norm bound is computed numerically.

Fig. \ref{fig:ct_cart_pole_bounds} shows a similar result, but for the cart pole system. Since the system is higher dimensional, direct multiple shooting is used to compute the optimal control \cite{diehl_fast_2006}. This method does not approximate the value function, so we do not evaluate its quadratic bounds; however, solution curves satisfy the theoretical bounds.

\section{Discrete Time Stability}
We begin by setting up the discrete time problem with the same cost function form as the continuous time problem. After proving stability of this nominal formulation, we progress to a practical formulation closer to what is done in practice with RL. Then lastly, we examine the robustness properties of the closed loop controller.
\subsection{Nominal Formulation}
We define the stage cost:
\begin{equation}
    \ell(x, u) = \beta V(x) + \rho[\Delta V(x, u) + \lambda V(x)]_+,
\end{equation}
and formulate the discounted optimal-control problem as
\begin{equation*}
    J_{\pi}(x_0) = \sum_{k = 0}^\infty \delta^k \ell(x_k^\pi, \pi(x_k^\pi))
\end{equation*}
where $\delta \in (0,1)$ is the discount factor.

\begin{assumption}
\label{assump:dt_clf_existance}
    There exists constants $c_1$, $c_2$, $\alpha > 0$ and an admissible feedback $\mu : \mathcal{X} \rightarrow \mathcal{U}$ such that, for all $x \in \mathcal{X}$,
    \begin{align*}
        c_1 \|x\|^2 &\leq V(x) \leq c_2 \|x\|^2 \\
        V(F(x, \mu(x))) &- V(x) \leq -\alpha V(x).
    \end{align*}
\end{assumption}

\begin{assumption}
\label{assump:dt_optimal_reg}
    The optimal value function $J^*$ is finite and continuous on $\mathcal{X}$, and the optimal control problem admits a minimizing control policy $\pi^*: \mathcal{X} \rightarrow \mathcal{U}$. We denote the resulting closed loop dynamics as
    \begin{equation*}
        F^*(x) := F(x, \pi^*(x)).
    \end{equation*}
\end{assumption}

We define the sublevel sets of the value function $S_d$ and Lyapunov function $\Omega_c$ corresponding to continuous time.

\begin{theorem}
\label{thm:dt_exp_stability}
    Suppose assumptions \ref{assump:dt_clf_existance} and \ref{assump:dt_optimal_reg} hold. Then there exists constants $d > 0$ and $c > 0$ such that $S_d$ is compact and forward invariant for the optimal closed loop system. Further $\Omega_c \subset S_d$ and, for all $x \in \Omega_c$:
    \begin{align}
    \label{eq:dt_J_bounds}
        \beta c_1 \| x\|^2 \leq J^*(x) &\leq \frac{\beta c_2}{1 - \delta(1 - \lambda)} \|x\|^2 \\
    \label{eq:dt_J_dec}
        J^*(F^*(x)) &\leq (1 - \lambda)J^*(x).
    \end{align}
    The origin is locally exponentially stable under the optimal policy $\pi^*$. In particular,
    \begin{equation}
    \label{eq:dt_state_bound}
        \|x_k\| \leq \sqrt{\frac{c_2}{c_1(1 - \delta(1 - \lambda))}}(1 - \lambda)^{k/2}\|x_0\| \quad \forall k \in \mathbb{Z}_{\geq 0}.
    \end{equation}
\end{theorem}
\begin{proof}
    We start with the upper and lower bound on $J^*$. The lower bound can be seen as follows
    \begin{equation*}
        J^*(x) \geq \inf_{u \in \mathcal{U}}\ell(x, u) \geq \beta V(x) \geq \beta c_1 \|x\|^2
    \end{equation*}
    
    To show the upper bound, we leverage the CLF feedback $\mu$ from Assumption \ref{assump:dt_clf_existance}:
    \begin{equation*}
        \Delta V(x, \mu(x)) \leq -\alpha V(x) \leq -\lambda V(x)
    \end{equation*}
    and therefore the decreasing term in the stage cost is zero. From this we bound the Lyapunov value over the trajectory
    \begin{equation*}
        V(x_k^\mu) \leq (1 - \alpha)^kV(x) \leq (1 - \lambda)^kV(x).
    \end{equation*}
    Which can be used to bound the value function
    \begin{align*}
        J^*(x) &\leq J_\mu(x) = \sum_{k = 0}^\infty \delta^k \beta V(x_k^\mu) \\
        &\leq \beta \sum_{k = 0}^\infty \delta^k (1 - \lambda)^k V(x) = \frac{\beta}{1 - \delta(1 - \lambda)}V(x) \\
        &\leq \frac{\beta c_2}{1 - \delta (1 - \lambda)}\|x\|^2
    \end{align*}
    which proves \eqref{eq:dt_J_bounds}. These bounds show that $J^*$ is positive definite on $\mathcal{X}$. Hence $m_r$ (using the same definition in \eqref{eq:min_on_sphere}, but on the discrete time value function) is still positive for $r > 0$. The sublevel set $S_d$ is contained in $B_r$ and hence compact. Since $V$ is continuous and $V(0) = 0$ there exists a $c > 0$ such that $\Omega_c \subset S_d$.

    Next, we evaluate the Bellman equation on $\pi^*$:
    \begin{equation}
    \label{eq:dt_bellman_optimal}
        J^*(x) = \ell(x, \pi^*(x)) + \delta J^*(F^*(x)).
    \end{equation}
    Since $\ell(x, u) \geq \beta V(x)$, the upper bound in \eqref{eq:dt_J_bounds} implies
    \begin{equation}
    \label{eq:dt_stage_cost_lb}
        \ell(x, \pi^*(x,u)) \geq (1 - \delta(1 - \lambda))J^*(x).
    \end{equation}
    Substituting \eqref{eq:dt_stage_cost_lb} into \eqref{eq:dt_bellman_optimal} yields the decreasing condition on the value function
    \begin{equation*}
        J^*(F^*(x)) \leq (1 - \lambda)J^*(x),
    \end{equation*}
    which proves \eqref{eq:dt_J_dec}. So if $J^*(x) \leq d$ then this implies $J^*(F^*(x)) \leq d$ and therefore $S_d$ is forward invariant.

    Finally, let $\{x_k\}_{k \geq 0}$ be any solution of the optimal closed loop dynamics with $x_0 \in \Omega_c$. Iterating \eqref{eq:dt_J_dec} gives
    \begin{equation*}
        J^*(x_k) \leq (1 - \lambda)^kJ^*(x_0).
    \end{equation*}
    Then we combine these with the quadratic bounds to bound the state norm over the trajectory
    \begin{align*}
        \|x_k\|^2 &\leq \frac{1}{\beta c_1}J^*(x_k) \\
        &\leq \frac{1}{\beta c_1}(1 - \lambda)^k J^*(x_0) \\
        &\leq \frac{1}{\beta c_1}(1 - \lambda)^k \frac{\beta c_2}{1 - \delta(1 - \lambda)}\|x_0\|^2.
    \end{align*}
    Taking square roots yields \eqref{eq:dt_state_bound}. Therefore the origin is locally exponentially stable under the optimality policy.
\end{proof}

\begin{figure*}
    \centering
    \includegraphics[width=1.0\linewidth]{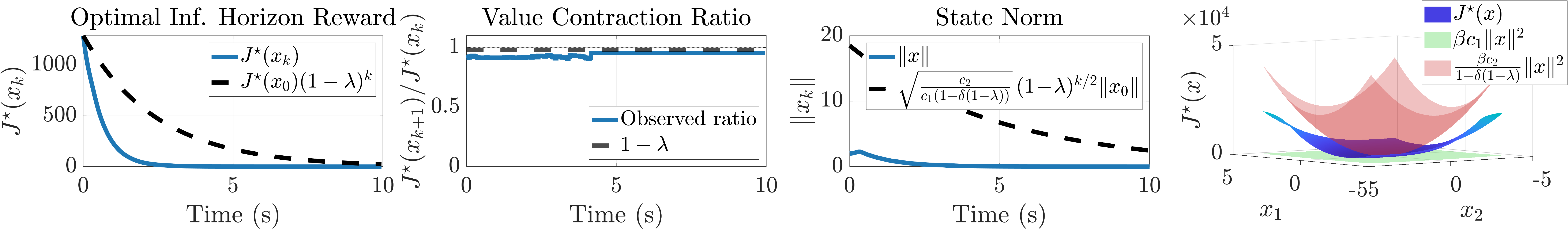}
    \caption{Theoretical bounds and the numerical solution to the discrete time double integrator system. The optimal control problem was solved numerically, yielding the value function surface and the optimal controller.}
    \label{fig:dt_double_integrator}
    \vspace{-5mm}
\end{figure*}

We verify the bounds derived here through a numerical example. Fig. \ref{fig:dt_double_integrator} shows the theoretical bounds plotted against numerical solutions for the discrete time double integrator. We solve this optimal control problem via dynamic programming, yielding an approximation of the value function which falls within the quadratic bounds.

\subsection{Practical Formulation}
Now we consider a formulation of the reward as done in the past for actual implementation with RL training  \cite{li_clf-rl_2026}. We add an additional, cost term, $r_{reg}$ which serves to regularize the resulting output. The form of the other rewards is also changed to better suite the numerics of RL training:
\begin{align*}
    r_V(x) &:= \beta \exp\left(-\frac{V(x)}{\sigma^2}\right) \\
    r_{\dot{V}}(x, u) &:= -\rho \: \text{clip}\left( \frac{\Delta V + \lambda V(x)}{\sigma_{\dot{V}}}, 0, 1\right) \\
    r_{reg} &:= -w_u \|u\|^2 \\
    \text{clip}(s, 0, 1) &:= \min\{\max\{s, 0\}, 1\}.
\end{align*}
The regularization terms could include other terms. We treat only an input penalty for simplicity.
The total reward is
\begin{equation*}
    R(x,u) = r_V + r_{\dot{V}} + r_{reg}
\end{equation*}

Now we convert the RL reward maximization problem to a cost minimization problem.
Noting that the maximum achievable reward is $\beta$, we define the stage cost and discounted cost over the horizon as
\begin{align*}
    \ell(x, u) &:= \beta - R(x, u) \\
    &\phantom{:}= \phi(V(x)) - r_{\dot{V}}(x, u) - r_{reg}(x, u) \\
    J(x_0) &:= \sum_{k = 0}^\infty \delta^k \ell(x_k^\pi, \pi(x_k^\pi))
\end{align*}
where
\begin{equation*}
    \phi(s) := \beta (1 - e^{s/\sigma^2}) \quad s \geq 0 \\
\end{equation*}

The shifted exponential $\phi(s)$ is locally equivalent to the linear state cost, as we now establish with a Lemma.

\begin{lemma}
\label{lem:phi_equivalence}
    Fix $\bar{c} > 0$ and define the CLF sublevel set
    \begin{equation*}
        \Omega_{\bar{c}} := \{x \in \mathcal{X} : V(x) \leq \bar{c}\}.
    \end{equation*}
    Then for every $x \in \Omega_{\bar{c}}$,
    \begin{equation*}
        \frac{\beta}{\sigma^2}e^{-\bar{c}/\sigma^2} V(x) \leq \phi(V(x)) \leq \frac{\beta}{\sigma^2} V(x).
    \end{equation*}
\end{lemma}
\begin{proof}
    Since $1 - e^{-s} \leq s$ for all $s \geq 0$, the upper bound is:
    \begin{equation*}
        \phi(V(x)) = \beta (1 - e^{-V(x)/\sigma^2}) \leq \frac{\beta}{\sigma^2} V(x).
    \end{equation*}
    For the lower bound, leverage the following integral bound
    \begin{align*}
        1 - e^{-s} &= \int_0^s e^{-\tau} d\tau \geq e^{-\bar{c}/\sigma^2}s \\
        \phi(V(x)) &= \beta(1 - e^{-s}) \geq \frac{\beta}{\sigma^2}e^{-\bar{c}/\sigma^2}V(x)
    \end{align*}
    establishing the lower bound.
\end{proof}

Moving forward, denote the coefficients derived above $\zeta_{-}(\bar{c}) := \frac{\beta}{\sigma^2}e^{-\bar{c}/\sigma^2}$ and $\zeta_{+} := \frac{\beta}{\sigma^2}$.

\begin{assumption}
\label{assump:dt_regularizer_domination}
    There exists $\bar{c}, c_{reg} \geq 0$ such that on $\Omega_{\bar{c}}$,
    \begin{equation*}
        0 \leq -r_{reg}(x, \mu(x)) \leq c_{reg} V(x)
    \end{equation*}
    where $\mu$ is the stabilizing feedback from Assumption \ref{assump:dt_clf_existance}.
\end{assumption}

This assumption states that the regularization terms are dominated by the Lyapunov function value. Although here we are only explicitly considering $\|u\|$ as a regularization term, this condition can be satisfied for more general terms. 

Now we can state the practical reward stability result.
\begin{theorem}
\label{thm:dt_practical_stab}
    Suppose Assumptions \ref{assump:dt_clf_existance} and \ref{assump:dt_regularizer_domination} hold and
    \begin{equation*}
        q_{\bar{c}} := \delta^{-1}\left( 1 - \frac{(1 - \delta(1 - \lambda))\zeta_{-}(\bar{c})}{\zeta_{+} + c_{reg}} \right) < 1.
    \end{equation*}
    Then for all $x \in \Omega_c$:
    \begin{align}
    \label{eq:dt_practical_J_bounds}
        \zeta_{-}(\bar{c})c_1\|x\|^2 \leq J^*(x) &\leq \frac{(\zeta_{+} + c_{reg})c_2}{1 - \delta(1 - \lambda)} \|x\|^2,\\
        \label{eq:dt_practical_J_dec}
        J^*(F^*(x)) &\leq q_{\bar{c}}J^*(x).
    \end{align}
    Therefore the origin is a fixed point of the optimal closed loop system and is exponentially stable. In particular, every optimal trajectory satisfies, for all $\forall k \in \mathbb{Z}_{\geq 0}$
    \begin{equation}
    \label{eq:dt_practical_state_bounds}
        \|x_k\| \leq \sqrt{\frac{(\zeta_+ + c_{reg})c_2}{\zeta_{-}(\bar{c})c_1(1 - \delta (1 - \lambda))}}q_{\bar{c}}^{k/2}\|x_0\|
    \end{equation}
    when $x_0 \in \Omega_c$.
\end{theorem}
\begin{proof}
    Once again we have $\Delta V(x, \mu(x)) \leq -\lambda V(x)$. From Lemma \ref{lem:phi_equivalence} and Assumption \ref{assump:dt_regularizer_domination} we get
    \begin{equation*}
        \ell(x, \mu(x)) \leq (\zeta_+ + c_{reg}) V(x).
    \end{equation*}
    Iterating the inequality on $V(x)$ yields $V(x_k^\mu) \leq (1 - \lambda)^kV(x)$, and therefore we can obtain the upper bound:
    \begin{align}
        J^*(x) &\leq \sum_{k = 0}^\infty \ell(x_k^\mu, \mu(x_k^\mu)) \nonumber \\
        &\leq (\zeta_{+} + c_{reg})\sum_{k = 0}^\infty \delta^k(1 - \lambda)^kV(x) \nonumber \\
        \label{eq:dt_practical_thm_J_upper_bound}
        &= \frac{\zeta_+ + c_{reg}}{1 - \delta(1 - \lambda)}V(x) \leq \frac{(\zeta_+ + c_{reg})c_2}{1 - \delta(1 - \lambda)}\|x\|^2
    \end{align}
    The lower bound comes immediately from the stage cost:
    \begin{equation*}
        J^*(x) \geq \inf_{u \in \mathcal{U}}\ell(x, u) \geq \phi(V(x)) \geq \zeta_{-}(\bar{c})c_1\|x\|^2.
    \end{equation*}
    To establish the decreasing condition \eqref{eq:dt_practical_J_dec}, consider the Bellman equation \cite{sutton1999reinforcement}:
    \begin{equation*}
        J^*(x) = \ell(x, \pi^*(x)) + \delta J^*(F^*(x)).
    \end{equation*}
    Then using $\ell \geq \phi(V(x)) \geq \zeta_{-}(\bar{c})V$ on $\Omega_{\bar{c}}$,
    \begin{equation*}
        \delta J^*(F^*(x)) \leq J^*(x) - \zeta_{-}(\bar{c}) V(x).
    \end{equation*}
    Now using the upper bound established in \eqref{eq:dt_practical_thm_J_upper_bound} we obtain
    \begin{equation*}
        V(x) \geq \frac{1 - \delta(1 - \lambda)}{\zeta_+ + c_{reg}}J^*(x)
    \end{equation*}
    and therefore $J^*(F^*(x)) \leq q_{\bar{c}}J^*(x)$. Finally, upper bound the value function using the CLF controller $\mu$:
    \begin{align*}
       J^*(x) &\leq \sum_{k = 0}^\infty \delta^k (\phi(V(x_k^\mu)) - r_{reg}(x_k^\mu, \mu(x_k))) \\
       &\leq \sum_{k=0}^\infty \delta^k(\zeta_{+} + c_{reg})V(x_k^{\mu}) \\
       &\leq \sum_{k=0}^\infty \delta^k(\zeta_{+} + c_{reg})(1 - \lambda)^k V(x)\\
       &= \frac{\zeta_{+} + c_{reg}}{1 - \delta(1 - \lambda)}V(x) \leq \frac{(\zeta_{+} + c_{reg})c_2}{1 - \delta(1 - \lambda)}\|x\|^2
    \end{align*}
    
    Using the same arguments at Thm. \ref{thm:dt_exp_stability} there exists a $d$ and $c$ such that $S_d$ is compact and forward invariant and $\Omega_c \subset S_d \subset \Omega_{\bar{c}}$. Finally, we can iterate \eqref{eq:dt_practical_J_dec} and combine with \eqref{eq:dt_practical_J_bounds} to get \eqref{eq:dt_practical_state_bounds}.
\end{proof}
\subsection{Robustness}
Since the resulting policy is exponentially stabilizing, we can apply standard input-to-state stability arguments to guarantee robustness to model mismatch and bounded suboptimality. This holds for both discrete time formulations.
\begin{definition}
    \cite{geiselhart2016relaxed} A discrete time dynamical system 
    \begin{equation*}
        x_{k+1} = F(x_k, d_k)
    \end{equation*}
    is exponentially input to state stable (E-ISS) with respect to a uniformly bounded disturbance $d$ if there exist constants $M, \lambda > 0$ and a class $\mathcal{K}$ function $\alpha \in \mathcal{K}$ such that:
    \begin{equation*}
        \|x_k\| \leq M \lambda^k + \alpha(\|d\|_\infty)
    \end{equation*}
\end{definition}
\begin{definition}
    A function $V: \mathcal{X} \to \mathcal{R}$ is an E-ISS Lyapunov function if there exist constants $k_1, k_2 > 0$, $k_3 \in (0, 1)$, and $\sigma \in \mathcal{K}$ such that
    \begin{align*}
        k_1\|x\|^2 &\leq V(x) \leq k_2 \|x\|^2 \\
        V(F(x, d)) &\leq k_3 V(x) + \sigma(d)
    \end{align*}
\end{definition}
If there exists an E-ISS Lyapunov function for the system, then the system is E-ISS \cite{geiselhart2016relaxed}. Now, we consider an additive disturbance to the dynamics.
\begin{corollary}
    Consider the addition of a bounded state-dependent disturbance to the dynamics
    \begin{equation*}
        x_{k+1} = F(x_k, u_k) + d_k
    \end{equation*}
    with $\|d_k\| \leq \bar{d}$. Under the action of the optimal controller, the system is exponentially ISS to the origin.
\end{corollary}
\begin{proof}
    From Theorem \ref{thm:dt_exp_stability} and \ref{thm:dt_practical_stab}, we have that $J^*(x)$ satisfies quadratic bounds as well as the difference condition:
    \begin{equation} \label{eqn:disc_opt_diff}
        J^*(F(x, \pi^*(x))) \leq c_3 J^*(x)
    \end{equation}
    where we use $c_3$ to denote the coefficient in either \eqref{eq:dt_J_dec} or \eqref{eq:dt_practical_J_dec}.
    By continuity of $J^*$, for any $x$ there must exist $\sigma_x \in \mathcal{K}$:
    \begin{equation}\label{eq:val_cont}
        |J(x + d) - J(x)| \leq \sigma_x(\|d\|)
    \end{equation}
    We uniformly bound over the compact sublevel set $\Omega_c$:
    \begin{equation*}
        \sigma(s) = \sup_{x \in \Omega_c} \sigma_x(s)
    \end{equation*}
    where $\sigma \in \mathcal{K}$. Applying to \eqref{eq:val_cont} and rearranging:
    \begin{equation*}
        J(x + d) \leq J(x) + \sigma(d)
    \end{equation*}
    Finally, examine effect of the disturbance on \eqref{eqn:disc_opt_diff}:
    \begin{align*}
        J^*(F(x, \pi^*(x)) + d) &\leq J(F(x, \pi^*(x))) + \sigma(d) \\
        &\leq c_3J(x) + \sigma(d).
    \end{align*}
    $J^*$ is an E-ISS Lyapunov function.
\end{proof}

\begin{corollary} \label{crl:robust}
    Consider a controller $\Tilde{\pi}(x)$ which has been learned to approximate $\pi^*(x)$ to bounded suboptimality, i.e.
    \begin{equation*}
        \|\pi^*(x) - \Tilde{\pi}(x)\| \leq \bar{d}
    \end{equation*}
    Assume that the closed loop dynamics under the learned controller, $F(x_k, \Tilde{\pi}(x_k))$ are locally Lipschitz continuous.
    Under the action of $\Tilde{\pi}$, the system is exponentially ISS to the origin.
\end{corollary}
\begin{proof}
    Restructuring the error in the controller:
    \begin{align*}
    x_{k+1} &= F(x_k,\tilde{\pi}(x_k)) \\
            &= F(x_k,\pi^*(x_k)) + F(x_k,\tilde{\pi}(x_k)) \nonumber \\
            &\phantom{{}={}} \phantom{F(x_k,\pi^*(x_k))} {}- F(x_k,\pi^*(x_k)) \\
            &= F(x_k,\pi^*(x_k)) + d_k
    \end{align*}
    Since the closed loop dynamics are Lipschitz continuous, we have 
    \begin{align*}
        \|d_k\| &= \|F(x_k, \Tilde{\pi}(x_k)) - F(x_k, \pi^*(x_k))\| \\
        &\leq L_F \|\pi^*(x) - \Tilde{\pi}(x)\| \leq L_F \bar{d}
    \end{align*}
    This system can be modeled as a bounded additive disturbance, so E-ISS can be concluded.
\end{proof}
\subsection{RL Results}
We can now take these concepts and apply them to a real RL problem. In the following examples we train a policy using IsaacLab \cite{nvidia_isaac_2025} for the environment and PPO \cite{rudin_learning_2022} for the training algorithm. We deploy standard RL training tools with the CLF-RL rewards analyzed by this paper.

To illustrate the practical bounds derived above, we start with a double integrator model attempting to stabilize to the origin. Fig. \ref{fig:rl_double_int} shows the practical state bound plotted against the closed loop RL policy. We can see that the practical bounds are verified experimentally and the double integrator is exponentially stable.
\begin{figure}
\vspace{2mm}
    \centering
    \includegraphics[width=1.0\linewidth]{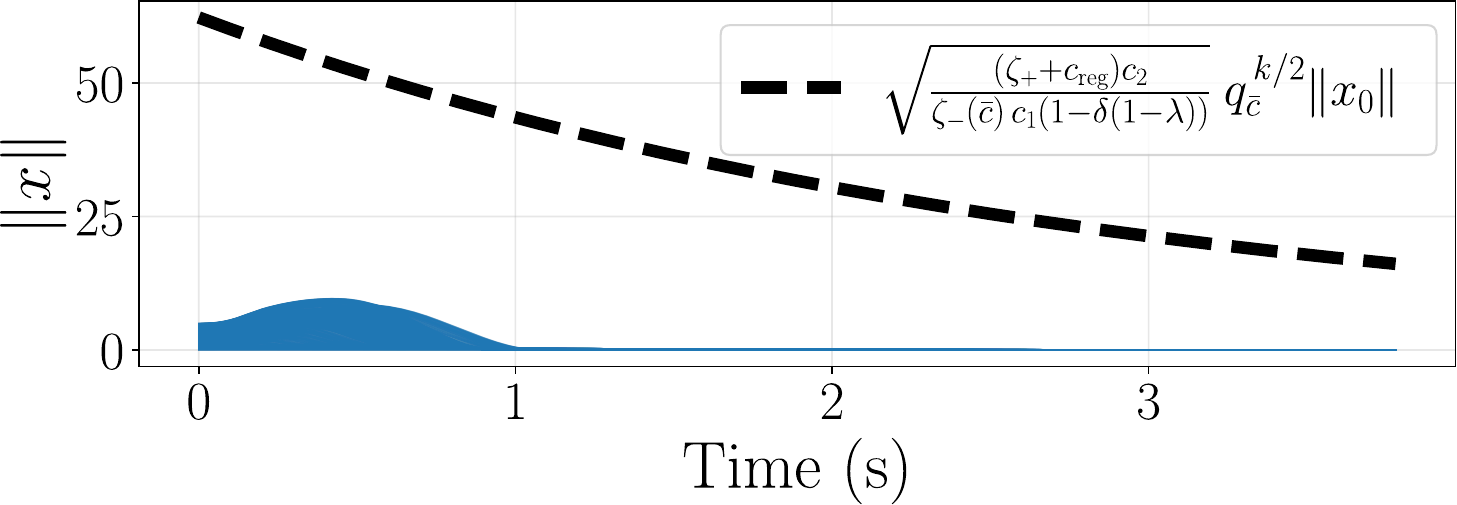}
    \caption{Results of the double integrator control policy trained with RL. The plot shows the theoretical bound from \eqref{eq:dt_practical_state_bounds} over 2000 sampled initial conditions for the CLF-RL policy. We can see that the theoretical bounds are satisfied even for the real RL problem.
    }
    \label{fig:rl_double_int}
    \vspace{-5mm}
\end{figure}

Finally, we apply the CLF-RL rewards to a full order humanoid system to learn a policy that can achieve stable walking. A Unitree G1 robot learns to track an optimized periodic gait (see \cite{li_clf-rl_2026} for more details on the training routine, although note that in this case holonomic rewards are not used). Fig. \ref{fig:rl_humanoid} shows the tracking of the resulting policy. We can see the state exponentially approach the orbit from a standing configuration. The bottom plot shows the state converge towards the orbit and the value function converge towards its maximum. 

\begin{figure}
    \vspace{2mm}
    \centering
    \includegraphics[width=1.0\linewidth]{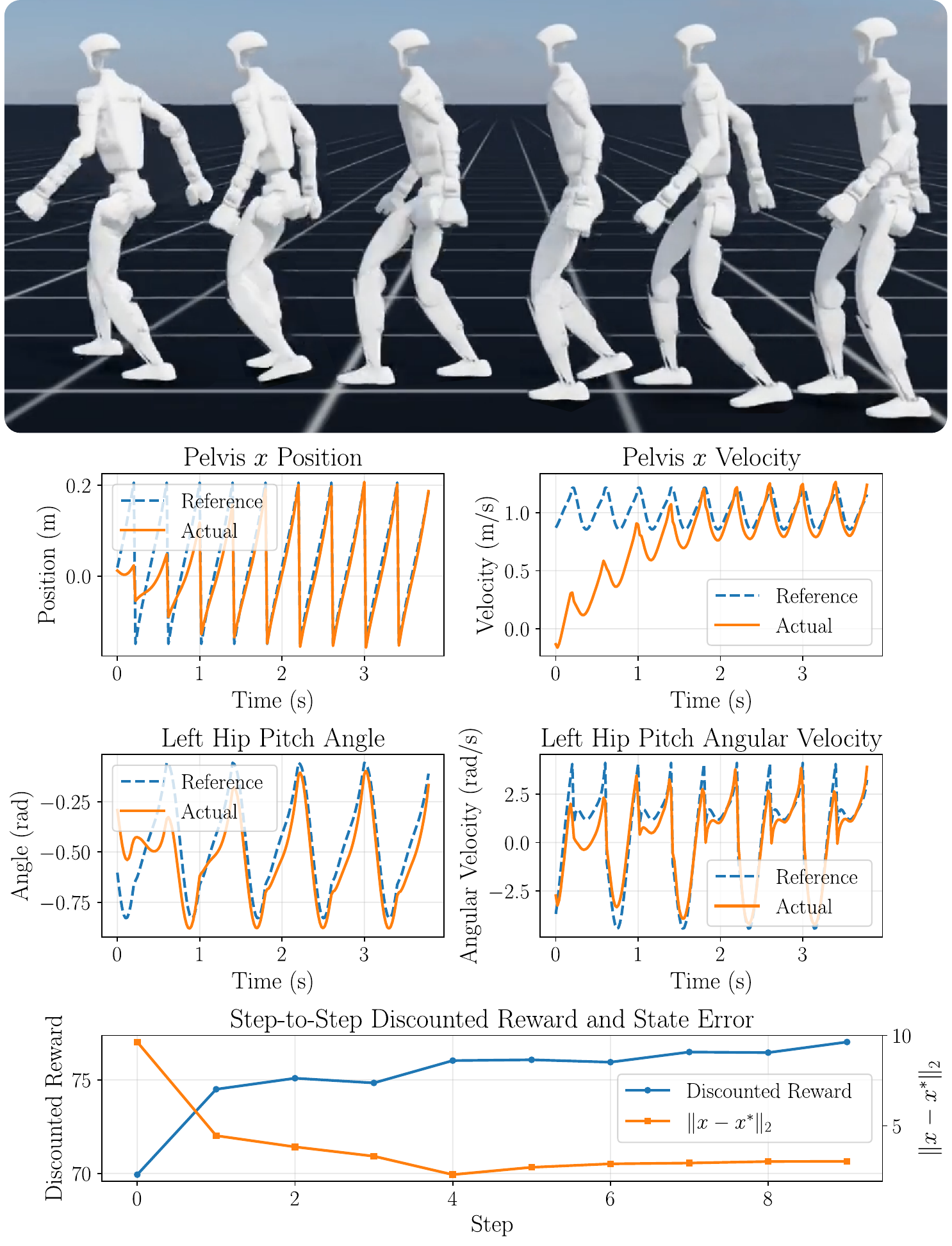}
    \caption{CLF-RL applied to a humanoid robot. Rather than tracking a single point, the CLF tracks a reference trajectory which forms a periodic orbit in the state space. We plot the tracking of the position and velocity of the pelvis and the left hip pitch joint. We also plot the discounted reward and normed error on foot impact to show the discrete step-to-step system's convergence. The robot starts from standing still and converges into the walking orbit.}
    \label{fig:rl_humanoid}
    \vspace{-6mm}
\end{figure}

\section{Conclusion and Future Work}
In this contribution we proved exponential stability of optimal controllers that used the CLF-RL rewards. We viewed the RL problem as an optimal control problem and examined the reward formulations in both discrete and continuous time. In each case the theoretical bound is confirmed by numerical examples. This grounds a practically successful method in theory, strengthening the decision to use these rewards. By analyzing this reward shaping methodology we help to reduce the number of heuristics and arbitrary choices made during the RL engineering process. Further, the CLF-RL method was applied to control a humanoid robot, resulting in stable periodic walking.

Although this work takes the first step in theoretically verifying the stability of the CLF-RL rewards, there are a number of next steps to take. The optimization in RL is done with stochastic control actions, i.e. as in PPO, and the mean action is deployed. This could be analyzed in the above context for stability guarantees. Further, the policy class searched over in RL (neural network representation, action parameterization, and control frequency) was not considered in this work. Adding these concepts would help bridge the gap between the theory and the application. For the humanoid specifically, the hybrid nature of the system should be considered and stability to a point would be replaced by stability to an orbit (i.e. steady state walking). These ideas can build on the ground work presented here.

\bibliographystyle{IEEEtran}
\bibliography{CLF-RL_theory}


\end{document}